\documentclass[11pt]{article}
\usepackage[utf8]{inputenc}
\usepackage[T1]{fontenc}
\usepackage[sc]{mathpazo}
\usepackage{microtype}
\usepackage{amsmath,xcolor}
\usepackage{amssymb}
\usepackage{amsthm}
\usepackage{bm,float}
\usepackage{dsfont}
\usepackage{authblk}
\usepackage{fullpage,caption,wrapfig}
\usepackage{comment}
\usepackage{mathtools}
\usepackage[shortlabels]{enumitem}
\usepackage{complexity}
\usepackage[backend=bibtex, style=alphabetic, backref=true,maxbibnames=99, url=false]{biblatex} 
\usepackage{bbm}
\usepackage{hyperref}
\hypersetup{
    colorlinks=true,
    linkcolor=violet,
    filecolor=magenta,      
    urlcolor=cyan,
    citecolor=blue,
    pdffitwindow=true,
}
\usepackage[capitalize, nameinlink]{cleveref}
\usepackage{nag,tikz}
\usetikzlibrary{calc}
\usetikzlibrary{decorations.pathreplacing}

\usepackage{algorithm,algpseudocode}
\usepackage{multicol}
\usepackage{thmtools}

\theoremstyle{plain}
\newtheorem{theorem}{Theorem}[section]
\newtheorem{lemma}[theorem]{Lemma}
\newtheorem{corollary}[theorem]{Corollary}

\newtheorem{fact}[theorem]{Fact}
\newtheorem{conjecture}[theorem]{Conjecture}

\newtheorem{definition}[theorem]{Definition}

\theoremstyle{remark}
\newtheorem{remark}[theorem]{Remark}

\newcommand{\old}[1]{}

\renewcommand{\R}{\ensuremath{\mathbb R}}
\newcommand{\Z}{\ensuremath{\mathbb Z}}

\renewcommand{\P}[1]{{\mathbb{P}}\left[#1\right]}
\renewcommand{\PP}[2]{{\mathbb{P}}_{#1}\left[#2\right]}

\renewcommand{\E}[1]{{\mathbb{E}}\left[#1\right]}
\renewcommand{\EE}[2]{{\mathbb{E}}_{#1}\left[#2\right]}

\newcommand{\Var}[1]{{\operatorname{Var}}\left[#1\right]}
\newcommand{\VV}[2]{{\operatorname{Var}}_{#1}\left[#2\right]}

\def\b1{{\bf 1}}
\def\1{{\bf 1}}

\def\eps{{\epsilon}}

\def\R{\mathbb{R}}

\definecolor{mygreen}{RGB}{85,170,102}
\definecolor{myblue}{RGB}{57,121,255}
\definecolor{mypurple}{RGB}{121,57,216}
\definecolor{mygray}{RGB}{180,180,180}

\newcommand{\declareperson}[1]{\expandafter\newcommand\csname#1\endcsname[1]{\textcolor{orange}{#1: ##1}}}

\declareperson{Anna}
\declareperson{Shayan}
\declareperson{Nathan}
\declareperson{Xinzhi}
\begin{document}
\title{An Improved Approximation Algorithm for the Minimum $k$-Edge Connected Multi-Subgraph Problem}
\author{Anna R. Karlin\thanks{\href{mailto:karlin@cs.washington.edu}{karlin@cs.washington.edu}. Research supported by Air Force Office of Scientific Research grant FA9550-20-1-0212 and NSF grant CCF-1813135.}}
\author{Nathan Klein\thanks{\href{mailto:nwklein@cs.washington.edu}{nwklein@cs.washington.edu}. Research supported in part by NSF grants DGE-1762114, CCF-1813135, and CCF-1552097.}}
\author{Shayan Oveis Gharan\thanks{\href{mailto:shayan@cs.washington.edu}{shayan@cs.washington.edu}. Research supported by Air Force Office of Scientific Research grant FA9550-20-1-0212, NSF grants  CCF-1552097, CCF-1907845,  
and a Sloan fellowship.}} 
\author{Xinzhi Zhang\thanks{\href{mailto:xinzhi20@cs.washington.edu}{xinzhi20@cs.washington.edu}. {Research supported by  NSF grant CCF-1813135.}}} 
\affil{University of Washington}
\maketitle

\begin{abstract}
We give a randomized $1+\frac{5.06}{\sqrt{k}}$-approximation algorithm for the minimum $k$-edge connected spanning multi-subgraph problem, $k$-ECSM.
\end{abstract}
\thispagestyle{empty} 

\newpage 
\setcounter{page}{1}

\section{Introduction}
In an instance of the minimum $k$-edge connected spanning subgraph problem, or $k$-ECSS, we are given an (undirected) graph $G=(V,E)$ with $n:=|V|$ vertices and a cost function $c:E\to\R_{\geq 0}$, and we want to choose a minimum cost set of edges $F\subseteq E$ such that the subgraph $(V,F)$ is $k$-edge connected. In its most general form, $k$-ECSS generalizes several extensively-studied problems in network design such as tree augmentation or cactus augmentation, for which there has been recent exciting progress (e.g. \cite{FGKS18,CTZ21,TZ21,BGA20}).
The $k$-edge-connected  {\em multi}-subgraph problem, $k$-ECSM, is a close variant of $k$-ECSS in which we want to choose a $k$-edge-connected {\em multi}-subgraph of $G$ of minimum cost, i.e., we can choose an edge $e\in E$ multiple times. Note that without loss of generality we can assume the cost function $c$ in $k$-ECSM is a metric, i.e., for any three vertices $x,y,z\in V$, we have $c(x,z)\leq c(x,y)+c(y,z)$.

Around four decades ago, Fredrickson and J\'aj\'a \cite{FJ81,FJ82} designed a 2-approximation algorithm for $k$-ECSS and a 3/2-approximation algorithm for $k$-ECSM. The latter essentially follows by a reduction to the well-known Christofides-Serdyukov approximation algorithm for the traveling salesperson problem (TSP).
Over the last four decades, despite a number of papers on the problem \cite{JT00,KR96, Kar99,Gab05,GG08,  GGTW09,Pri11,LOS12}, the aforementioned approximation factors were only improved in the cases where the underlying graph is unweighted or $k\gg \log n$. Most notably, Gabow, Goemans, Tardos and Williamson \cite{GGTW09} showed that if the graph $G$ is unweighted then $k$-ECSS and $k$-ECSM admit  $1+2/k$ approximation algorithms, i.e., as $k\to\infty$ the approximation factor approaches 1. 
The case of $k$-ECSM where $k=2$ has received significant attention and (significantly) better than $3/2$-approximation algorithms were designed for special cases \cite{CR98,BFS16,SV14,BCCGISW20}. In the general $k=2$ case, only a $3/2-\epsilon$ approximation is known where $\epsilon = 10^{-36}$ \cite{KKO21b}; we remark this also extends to all even $k$.

 Motivated by \cite{GGTW09},  Pritchard posed the following conjecture:
 \begin{conjecture}[\cite{Pri11}]\label{conj:kECSM}
The $k$-ECSM problem admits a $1+O(1)/k$ approximation algorithm.
 \end{conjecture}
 In other words, if true, the above conjecture implies that the 3/2-classical factor can be substantially improved for large $k$, and moreover that it is possible to design an approximation algorithm whose factor gets arbitrarily close to 1 as $k\to\infty$. In this paper, we prove a weaker version of the above conjecture.
 \begin{restatable}[Main]{theorem}{mainthm}\label{thm:main}
 There is a polynomial time randomized algorithm for (weighted) $k$-ECSM	 with approximation factor (at most)  $1+\frac{5.06}{\sqrt{k}}$.
 \end{restatable}
We remark that our main theorem only improves the classical 3/2-approximation algorithm for $k$-ECSM when $k > 103$.  However, the constants are not optimized and we expect our algorithm to beat $3/2$ for much smaller values of $k$. 



For a set $S\subseteq V$, let $\delta(S)=\{\{u,v\}: |\{u,v\}\cap S|=1\}$ denote the set of edges with one endpoint in $S$. The following is the natural linear programming relaxation for  $k$-ECSM.
\begin{equation}
\begin{aligned}
	\min \hspace{3ex}& \sum_{e\in E}x_e c(e) & \\
    \text{s.t.} \hspace{3ex}& x(\delta(v))=k &\forall v\in V\\
    &x(\delta(S)) \ge k & \forall S\subseteq V, S \neq \emptyset \\
    & x_e \ge 0 & \forall e\in E.
\end{aligned}
\label{eq:k_ecsmlp}
\end{equation}
Note that while in an optimum solution of $k$-ECSM the degree of each vertex is not necessarily equal to $k$, since the cost function satisfies the triangle inequality we may assume that in any optimum fractional solution each vertex has (fractional) degree $k$. This follows from the parsimonious property \cite{gb93}. 


We prove \cref{thm:main} by rounding an optimum solution to the above linear program. So, as a corollary we also upper-bound the integrality gap of the above linear program.
\begin{corollary}
	The integrality gap of LP \eqref{eq:k_ecsmlp} is at most $1+\frac{5.06}{\sqrt{k}}$.
\end{corollary}

\subsection{Proof Overview}\label{sec:overview}
Before explaining our algorithm, we recall a randomized rounding approach of Karger \cite{Kar99}. 
Karger showed that if given a solution $x$ to \eqref{eq:k_ecsmlp} we choose every edge $e$  independently with probability $x_e$, then the sample is $k-O(\sqrt{k\log n})$-edge connected with probability close to 1. 
He then fixes the connectivity of the sample by adding $O(\sqrt{k\log n})$ copies of the minimum spanning tree of $G$. This gives a randomized $1+O(\sqrt{\log n/k})$ approximation algorithm for the problem. While this is a very effective procedure for large $k$, it is not useful when $k$ is a constant or grows slower than $\log n$. We view our result as a refinement of this method using random spanning trees which allows $k$ to be independent of $n$.

First, we observe that when $x$ is a solution to \eqref{eq:k_ecsmlp}, the vector $2x/k$ is in the spanning tree polytope (after modifying $x$ slightly, see \cref{fact:opt_in_sp_polytope} for more details). Following a recent line of works on the traveling salesperson problem \cite{OSS11,KKO21} we write $2x/k$ as a so-called {\em max-entropy distribution} $\mu$ over spanning trees. 

\paragraph{Warm-up algorithm and key idea.} Our first algorithm, explained in \cref{sec:warmup}, independently samples $k/2$ spanning trees $T_1,\dots,T_{k/2}$ from $\mu$. 
Call the (multi-set) union of these trees $T^*$. Since max entropy distributions are negatively correlated, it is easy to show using Chernoff bounds that any particular cut $S$ has at least $k-O(\sqrt{k \ln k})$ edges with probability at least $1-O(1/\sqrt{k})$\footnote{Of course, one can make this probability much closer to 1 (say $1-O(1/k^2)$) by only paying a constant factor in the $O(\sqrt{k \ln k})$ term, but it is sufficient to make it $1-O(1/\sqrt{k})$.}. So, in the second step of the algorithm, we add $O(\sqrt{k \ln k})$ additional spanning trees to fix the connectivity of every cut ``with high probability." In other words, after this procedure (which has expected cost $1+O(\sqrt{\ln k/k})$ times the cost of the LP), \textit{every cut $S$ has at least $k$ edges with probability $1-O(1/\sqrt{k})$}. One can think of this as a version of Karger's algorithm which does not fix \textit{every cut} with high probability but instead fixes \textit{each individual cut} with high probability. 

A priori, this does not seem like a useful property, because there are exponentially many cuts to bound over. However, we show that (perhaps somewhat surprisingly) there is a way to fix the connectivity of every cut simultaneously by only paying an additional factor of $O(1/k)$ times the cost of the LP in expectation. 

To do so, we begin with the following simple observation: Fix a cut $S$. Then, if we ensure that every tree $T_i$ has at least 2 edges in $\delta(S)$, the union of the trees $T_i$ will have at least $k$ edges across the cut and we are done. So, if a cut $S$ turns out to have fewer than $k-O(\sqrt{k \ln k})$ edges in $T^*$, one can think of ``blaming" the trees which had only one edge in $\delta(S)$; in particular, we will fix the cut by doubling the sole edge in $\delta(S)$ for each of those trees. This guarantees that every tree has at least two edges across the cut and therefore it has at least $k$ edges total as desired. This is essentially the key idea of this paper. 

Formally, after sampling $T_1,\dots,T_{k/2}$, iterate over every edge $e$ of every tree $T_i$ and consider the unique cut $S$ in which $e$ appears as the only edge of $T_i$. This is the \textit{only} cut for which $e$ may be ``blamed" and hence doubled. Now, we check if $T^*$ has fewer than $k-O(\sqrt{k \ln k})$ edges across $\delta(S)$. Over the randomness of the remaining trees, this occurs with probability $O(1/\sqrt{k})$ (by Chernoff bounds, as argued above). This shows that every edge of $T^*$ is doubled with probability $O(1/\sqrt{k})$, and therefore since $T^*$ has expected cost at most $c(x) \le OPT$, the approximation ratio of the algorithm is $1+O(\sqrt{\ln k/k}) + O(1/\sqrt{k})$, i.e. $1+O(\sqrt{\ln k/k})$, as desired. 

\paragraph{Main algorithm.} The above algorithm is suboptimal in a fairly obvious way. Suppose that a cut $S$ is missing (for example) $\sqrt{k \ln^2 k}$ edges. Then, its connectivity is not fixed by the additional $O(\sqrt{k \ln k})$ spanning trees added by the algorithm. So, the warm-up algorithm simply adds an additional copy of every edge which appeared alone in this cut in its tree. However, it may be that the cut $\delta(S)$ has only one edge in as many as $\Omega(k)$ trees! Therefore, we will add $\Omega(k)$ extra edges to fix the cut instead of just the required $\sqrt{k \ln^2 k}$ edges: a huge overcorrection. \cref{alg:main} simply avoids this overcorrection by only adding the number of edges actually missing from the cut, sampling them independently from the set of edges which appeared alone on this cut. It turns out this will let us add only $O(\sqrt{k})$ additional trees instead of $O(\sqrt{k \ln k})$, avoiding the extra $\sqrt{\ln k}$ factor\footnote{For a slightly tighter analysis we also include these additional trees in $T^*$, but this is mostly a superficial difference.}.  

However, this adds some difficulty to the analysis. For example, now that we only add $O(\sqrt{k})$ extra trees, it is not true that a cut only has to be fixed with probability $O(1/\sqrt{k})$. In fact, this probability may even be $O(1)$. To sharpen the analysis, for any fixed set $S$ we study $p_S$, the probability that a random (max entropy) tree has exactly one edge in $\delta(S)$. In particular, we show that the expected number of edges that $\delta(S)$ is missing (below $k$) is at most $O(\sqrt{k}e^{-1/p_S})$. 

This bound on the expectation is then enough to complete the argument as follows. Let $n(S)$ be the number of trees $T_i$ for which $|T_i \cap \delta(S)| = 1$. Let $e$ be the unique edge in $\delta(S)$ for some $T_i$. 
Then, the probability $e$ needs to double is the expected number of edges missing from $\delta(S)$ divided by $n(S)$. Using the analysis above and that $\E{n(S)} = \Omega(kp_S)$, this can be shown to be $O(\frac{\sqrt{k}e^{-1/p_S}}{k p_S}) = O(1/\sqrt{k})$ in expectation\footnote{This is not immediate since this is the ratio of the expectations, but we actually need to analyze the expectation of the ratio.}.

\begin{remark}
We note that we expect this algorithm to work for any distribution of spanning trees which is negatively correlated. So, one could for example apply swap rounding \cite{CVZ10} to generate random spanning trees (instead of using the max entropy distribution). However, while the analysis giving $1+O(\sqrt{\frac{\ln k}{k}})$ approximation in \cref{sec:warmup} can easily be modified to give similar bounds for any negatively correlated distribution (since Chernoff bounds can be applied), the proof of \cref{thm:main} in \cref{sec:alg-analysis} currently relies on the fact that the distribution of the number of edges in any cut can be written as a sum of independent Bernoullis. So, an extension of \cref{thm:main} to an arbitrary negatively correlated distribution would require a different analysis technique or a generalization of \cref{lem:variance-bound}. 

We also briefly remark that some form of concentration is necessary. In particular, consider a distribution  over spanning trees in which a vertex $v$ has degree 1 with probability $1-1/(n-2)$ and degree $n-1$ with probability $1/(n-2)$. In such a case, we expect to need to add  $k/2$ edges from $\delta(v)$ to ensure $v$ has degree at least $k$. If these edges have all the cost of the LP (or there are many vertices with this property), we can get an approximation ratio as bad as 3/2 even for large $k$. 
\end{remark}

\section{Preliminaries}

\begin{definition}[$G^0,u_0,v_0$]\label{def:G0}
    We expand the graph $G=(V,E)$ to a graph $G^0$ by picking an arbitrary
vertex  $u\in V$, splitting it into two nodes $u_0$ and $v_0$, and then, for every edge $e= (u,w)$ incident to $u$, assigning fraction $\frac{x(e)}{2}$ to each of the two edges $(u_0,w)$ and $(v_0, w)$ in $G^0$ We set $x((u^0, v^0)) = 0$.  Call this expanded graph $G^0$, its edge set $E^0$, and the resulting fractional solution $x^0$, where $x^0(e)$ and $x(e)$ are identical on all other edges.  (Note that each of $u_0$ and $v_0$ now have fractional degree $k/2$ in $x^0$.) In \cref{fact:opt_in_sp_polytope} below, we show  that $\frac{2}{k} \cdot x^0$ is in the spanning tree polytope for the graph $G^0$.  For ease of exposition, the algorithm is described as running on  $G^0$ (and spanning trees\footnote{A spanning tree in $G^0$ is a 1-tree in $G$, that is, a tree plus an edge.} of $G^0$), which has the same edge set as $G$ (when $u_0$ and $v_0$ are identified).
\end{definition}

\subsection{Basic Notation}
For a subset of vertices $S \subseteq V$, we write $E(S) \subseteq E$ to denote the set of edges in the induced graph of $G$ whose vertex set is $S$.

For two sets of edges $F, F' \subseteq E$, we write $F \uplus F'$ to denote the multi-set union of $F$ and $F'$ allowing multiple edges. Note that we always have $|F\uplus F'|=|F|+|F'|$. 

For set of vertices $S\subseteq V$, let $\delta(S)= \{\{u,v\}: |\{u,v\}\cap S|=1\}$ denote the set of edges with one endpoint in $S$ and one endpoint in $\overline S$.

For any two sets of edges $F, T\subseteq E$, we write
\begin{align*}
    F_T := |F \cap T|.
\end{align*}
We will primarily use this notation to denote the number of edges $F$ has in a spanning tree (or union of spanning trees) $T$. 

Also, for any edge weight function $x: E \to \R$, we write $x(F) := \sum_{e\in F} x(e)$. 

\begin{definition}[$S_T(e)$, the ``One-Cut" of $e$ in $T$]\label{def:one-cut}
	For any spanning tree $T$ on the vertex set $V_0$, and any edge $e \in T$, let $S_T(e) \subseteq V_0 \smallsetminus \{u_0\}$ be the unique connected component of $T\smallsetminus \{e\}$ which does not contain $u_0$. We will call this the \textbf{one-cut} of $e$ in $T$.
\end{definition}
Particular edges $e\in T$ of interest are those where both $u_0,v_0\notin S_T(e)$.


Recall that the natural linear programming relaxation for $k$-ECSM is \eqref{eq:k_ecsmlp}.
The solution to this LP can be computed in polynomial time using the ellipsoid method.

For a real-valued random variable $X$, we write $X^+ = \max(0, X)$ to denote the positive part of $X$.

\subsection{Random Spanning Trees}\label{sec:prelim_st}
Edmonds \cite{Edm70} gave the following description for the convex hull of the spanning trees of any graph $G=(V,E)$, known as the {\em spanning tree polytope}.
\begin{equation}
\begin{aligned}
& z(E) = |V|-1 & \\
& z(E(S)) \leq |S|-1 &  \forall S\subseteq V\\
& z_e \geq 0 & \hspace{6ex} \forall e\in E.
\end{aligned}
\label{eq:spanningtreelp}
\end{equation}
Edmonds also \cite{Edm70} proved that the extreme point solutions of this polytope are the characteristic vectors of the spanning trees of $G$.

\begin{fact}[\cite{KKO21}]\label{fact:opt_in_sp_polytope}
    Let $x$ be the optimal solution of LP \eqref{eq:k_ecsmlp} and $x^0$ its extension to $G^0$ as described in \cref{def:G0}.
    Then $\frac{2}{k}\cdot x^0$ is in the spanning tree polytope  \eqref{eq:spanningtreelp} of $G^0$.
\end{fact}
\begin{proof}
For any set $S\subseteq V(G^0)$, we have $x^0(E(S)) = \sum_{v\in S}x^0(\delta(v)) - x^0(E(S))$.
If $u_0, v_0 \notin S$, then $x^0(\delta(v)) = k$ for all $v\in S$, and $x^0(\delta(S)) = x(\delta(S))$.
        Therefore, $$x^0(E(S)) =  \frac{k|S| - x(\delta(S))}{2} \leq \frac{k}{2}(|S| -1).$$
        
If $u_0 \in S, v_0\notin S$, then $\sum_{v\in V} x^0(\delta(v)) = k(|S|-1) + k/2 = k|S| - k/2$ and $x^0(\delta(S))\ge k/2$, so
  $$x^0(E(S)) \le  \frac{k|S| - k/2 - x^0(\delta(S))}{2} \leq \frac{k}{2}(|S| - 1).$$
Similar analysis also holds for the case where $u_0 \notin S, v_i \in S$.

Finally, if $u_0, v_0 \in S$, then $\sum_{v\in V} x^0(\delta(v)) = k(|S|-2) + 2 \cdot k/2 = k|S| - k$ and $x^0(\delta(S)) \geq k$. 
    Thus, $$x^0(E(S)) = \frac{k|S| - k - x^0(\delta(S))}{2} \leq \frac{k}{2}(|S| - 2).$$
The claim follows because 
        $x^0(E) =\frac{k|V(G)|}{2} = \frac{k}{2}(|V(G^0)| - 1).$ 
\end{proof}

Given nonnegative edge weights $\lambda: E \to \R_{\geq 0}$,
we say a distribution $\mu_\lambda$ over spanning trees of $G$ is \emph{$\lambda$-uniform}, if for any spanning tree $T$,
$$\PP{T \sim \mu_{\lambda}}{T} \propto \prod_{e\in T} \lambda(e).$$
It has been shown that the maximum entropy distribution over spanning trees is a $\lambda$-uniform distribution.

\begin{theorem}[\cite{AGMOS17}]\label{thm:max-entropy}
    There is a polynomial-time algorithm that, given a connected graph $G = (V,E)$, and a point $z\in \R^{|E|}$ in the spanning tree polytope \eqref{eq:spanningtreelp} of $G = (V,E)$, returns 
    $\lambda: E \to \R_{\geq 0}$ such that the corresponding $\lambda$-uniform spanning tree distribution $\mu_{\lambda}$ satisfies
    \begin{align*}
        \sum_{T\in \mathcal{T}: e\in T}\mu_{\lambda}(T) \leq (1+2^{-n}) z_e,~ \forall e\in E,
    \end{align*}
    i.e., the marginals are approximately preserved. In the above $\mathcal{T}$ is the set of all spanning trees of $G$.
\end{theorem}

\subsection{Bernoulli-Sum Random Variables}

In this section, we introduce several properties of the Bernoulli-sum random variable.
\begin{definition}[Bernoulli-Sum Random Variable]\label{def:bernoulli-rv}
    We say $BS(q)$ is a \emph{Bernoulli-Sum} random variable if it has the law of a sum of independent Bernoullis, say $B_1 + B_2 + \cdots + B_{t}$ for some $t \geq 1$, with $\E{B_1 + \cdots + B_t} = q$.
\end{definition}

\cite{BBL09, Pit97} showed that the size of any set of edges on any $\lambda$-uniform random spanning tree, i.e. for a set $F \subseteq E$, $F_T$ is distributed as a Bernoulli-sum random variable.

\begin{lemma}[Random variables $F_T$ are Bernoulli Sums \cite{BBL09, Pit97}]\label{lem:sum_bernoulli}
    Given $G = (V, E)$ and $\lambda : E\to \R_{\geq 0}$, let $\mu_\lambda$ be the $\lambda$-uniform spanning tree distribution of $G$. Let $T$ be a sample from $\mu_{\lambda}$. Then for any fixed $F \subseteq E$, the random variable $F_T$ is distributed as $BS(\E{F_T})$.
\end{lemma}

We start with a fact that comes directly from linearity of expectation and the definition of variance:
\begin{fact}\label{fact:sumBS}
    If $X=BS(q_1)$ and $Y=BS(q_2)$ are two independent Bernoulli-sum random variables, then $\E{X+Y} = q_1 + q_2$ and $\Var{X+Y} = \Var{X} + \Var{Y}$.	
\end{fact}

\begin{theorem}[Optimize Expected Value of BS Random Variable, \cite{Hoe56} Corollary 2.1]\label{thm:bs-optimize}
    Let $g:\{0,1,\cdots,m\}\to \R$ and $0\le p\le m$ for some integer $m > 0$. Let $X_1,\cdots,X_m$ be $m$ independent Bernoulli random variables with success probabilities $p_1,\cdots,p_m$ that minimizes (or maximizes)
    \begin{align*}
        \E{g(X_1 + \cdots + X_m)}
    \end{align*}
    such that $X_1 + \cdots + X_m = BS(p)$. Then $p_1,\cdots,p_m \in \{0,x,1\}$ for some $0<x<1$.
\end{theorem}

\begin{corollary}\label{cor:1/e}
For any $BS(q)$ with $q \ge 1$, $\P{BS(q) = 0} \le 1/e$.
\end{corollary}
\begin{proof}
    Suppose $BS(q) = X_1 + \cdots + X_m $ for some $m \in \Z_+$ where $X_1, \cdots, X_m$ are independent Bernoullis with success probability $p_1, \cdots, p_m$. 
    Let $g(x) = \mathbb{I}[x = 0]$.
    Then from \cref{thm:bs-optimize}, if we want to maximize $\P{BS(q) = 0} = \E{g(X_1 + \cdots + X_m)}$, then $p_1, \cdots, p_m \in \{0, x, 1\}$ for some $0 < x < 1$.
    Suppose $m_1$ of $p_i$'s are $1$, $m_2$ of $p_i$'s are $0$, and the rest of the $(m - m_1 - m_2)$ $p_i$'s are $\frac{q - m_1}{m - m_1 - m_2}$.
    Then we have 
    \begin{align*}
    	\Pr[BS(q) = 0] &= \prod_{i=1}^m (1 - p_i) \leq 0^{m_1} \cdot 1^{m_2} \cdot (1 - \frac{q - m_1}{m - m_1 - m_2})^{m - m_1 - m_2} \leq (1 - \frac{q}{m})^m \leq e^{-q}.
    \end{align*}
    where maximization is reached  when $m_1 = m_2 = 0$ and $m \to +\infty$. Notice that $q \geq 1$, we have $\Pr[BS(q) = 0] \leq 1/e$ as desired.
\end{proof}

\begin{fact}\label{fact:small-bern}
    Given any $0 \leq \epsilon < 1$,
    let $p_1 \ge p_2 \ge \dots p_m$ be the success probabilities of $m \geq 2$ independent Bernoullis such that $\sum_{i=1}^m p_i = 1 + \epsilon$. Suppose $p_1 \le \frac{1}{2}(1 + \epsilon)$. Then $\prod_{i=1}^m (1-p_i) \ge \frac{1}{4}(1 - \epsilon)^2$.
\end{fact}
\begin{proof}
The first step is to see that $\prod_{i=1}^m (1-p_i)$ is minimized when $p_1$ is as large as possible, i.e., $p_1=\frac12(1+\eps)$.
To see that, say $p_m>0$ (for some $m>1$) and observe that for any $0<\delta\leq p_m$,
$$ (1-(p_1+\delta)) (1-p_2)\dots(1-(p_m-\delta)) \leq \prod_{i=1}^m (1-p_i).$$
Note that this operation does not change the order of $p_i$'s. 
So, without loss of generality, assume $p_1=\frac12(1+\eps)$. Now, by Weierstrass inequality we have
\begin{align*}
	 \prod_{i=1}^m (1-p_i) \geq (1-p_1) \left(1-\sum_{i=2}^m p_i\right) &= (1-\frac12(1+\eps))(1-\frac12(1-\eps)) \geq \frac14(1-\eps)^2
\end{align*}

where the second to last identity uses that $\sum_i p_i = 1+\eps$. 
\end{proof}



\begin{theorem}[Bernstein Inequality for BS Random Variables]\label{thm:bernstein}
Let $X = BS(q)$ be a BS random variable with $\E{X}=q$ and $\Var{X} = \sigma^2$. Then $\forall \lambda > 0$ we have
$$\P{ X \leq q - \lambda} \le \exp \left( -\frac{\lambda^2}{2(\sigma^2 + \lambda / 3)} \right).$$
\end{theorem}

\begin{theorem}[Multiplicative Chernoff-Hoeffding Bound for BS Random Variables]\label{thm:chernoff}
Let $X=BS(q)$ be a Bernoulli-Sum random variable. Then, for any $0<\eps<1$ and $q' \le q$,
$$ \P{X<(1-\eps)q'}\leq e^{-\frac{\eps^2q'}{2}},$$
and for any $\epsilon > 0$, $q' \ge q$,
$$ \P{X>(1+\eps)q'}\leq e^{-\frac{\eps^2q'}{2+\epsilon}}.$$
\end{theorem}


\section{Warm-up: a Simple Algorithm with a $1 + O\big(\sqrt{\frac{\ln k}{k}}\big)$-Approximation Ratio.} \label{sec:warmup}

We first explain a simple algorithm (\cref{alg:warmup}) that has a slightly weaker $1 + O\big(\sqrt{\frac{\ln k}{k}}\big)$-approximation ratio.  We defer our main algorithm (\cref{alg:main}) and the proof of our main result (\cref{thm:main}) to \cref{sec:alg-analysis}.



\begin{algorithm}[]\caption{An Approximation Algorithm for $k$-ECSM}\label{alg:warmup}
\begin{algorithmic}[1]
    \State Let $x^0$ be an optimum solution of \eqref{eq:k_ecsmlp} extended to the graph $G^0$ as described above. \label{line:opt_ecss}
    \State Find  weights $\lambda : E^0 \to \R_{\ge 0}$ such that for any $e\in E^0$, $ \PP{\mu_{\lambda}}{e} \le  \frac{2}{k}x_e^0 (1 + 2^{-n})$.
    \Comment{By  \cref{thm:max-entropy}}
    \State Sample $k/2$ spanning trees $T_1,\cdots,T_{k/2}\sim\mu_\lambda$ (in $G^0$) independently and let $T^* \leftarrow T_1 \uplus \dots \uplus T_{k/2}$. 
    \label{line:sample_st}
    \State Let $B$ be the disjoint union of an additional $\alpha\sqrt{k/2-1}$\footnote{While the quantity $\sqrt{k/2-1}$ can of course be arbitrary, we choose it to make the analysis as simple as possible.} spanning trees sampled from $\mu_\lambda$.  
    \Comment{$\alpha=\Theta(\sqrt{\ln k})$ is a parameter we choose later.} 
    \label{line:baseline}
    \For {$i \in [\frac{k}{2}]$ and $e\in T_i$}
    		\If{$\delta(S_{T_i}(e))_{T^*} < k - \alpha\sqrt{k/2-1}$ and $(u_0,v_0)\notin \delta(S_{T_i}(e))$}
    			\State $F \leftarrow F \uplus \{e\}$. 
    			\label{line:augment}
    		\EndIf
    \EndFor
    \State {\bf Return} $T^* \uplus B \uplus F$.
\end{algorithmic}
\end{algorithm}

In the first step of \cref{alg:warmup}, we solve \eqref{eq:k_ecsmlp} on the (slightly) extended graph $G^0$. Let $x^0$ to be the optimal solution. By \cref{fact:opt_in_sp_polytope}, $(2/k)x^0$ is in the spanning tree polytope. Then in line 2, we find the $\lambda$-uniform spanning tree distribution $\mu_{\lambda}$ where each edge has marginal probability $(2/k)x^0_e$ (ignoring the $2^{-n}$ relative errors). This step is guaranteed to be done in polynomial time by \cref{thm:max-entropy}.

In line 3, we independently sample $k/2$ spanning trees\footnote{If $k$ is odd, we sample $\lceil{k/2}\rceil$ trees.  The bound remains unchanged relative to the analysis we give below as the potential cost of one extra tree is $O(OPT/k)$.} 
$T_1,\dots,T_{k/2}$ from $\mu_{\lambda}$, and let $T^* = T_1 \uplus \cdots \uplus T_{k/2}$ to be the (multi-set) union of the samples.
It follows that $T^*$ satisfies many desirable properties of the $\lambda$-uniform spanning tree distribution: 
\begin{enumerate}[i)]
	\item $T^*$ has the same expectation as the LP solution $x^0$, since the marginal probability of each edge is exactly $x^0(e)$;
	\item For any cut $\delta(S)$ in $G$, since $\delta(S)_{T^*}$ is distributed as a Bernoulli-sum random variable, Chernoff-type inequalities apply and $\delta(S)_{T^*}$ is highly concentrated around its mean;
	\item Since $T^*$ is the union of $k/2$ trees, for all cuts we have  $\delta(S)_{T^*} \geq k/2$. Moreover, if a cut $\delta(S)$ is not a tree cut of any of the $k/2$ trees, then each of the $k/2$ trees must have at least 2 edges crossing it. Therefore, the number of ``bad'' cuts of $T^*$, i.e. those with $\delta(S)_{T^*}< k $, is at most $(n-1)k/2$ (with probability 1).
\end{enumerate}

To fix the potentially $O(nk)$ bad cuts, we divide them into two types: (i) Cuts $S$ such that $\delta(S)_{T^*} \geq k-\alpha \sqrt{k/2-1}$ and (ii) Cuts $S$ where $\delta(S)_{T^*} < k-\alpha \sqrt{k/2-1}$, for some $\alpha=\Theta(\sqrt{\ln k})$.
We fix all cuts of type (i) by adding $B = \alpha \sqrt{k/2-1}$ additional spanning trees as in line 4 of the algorithm (note one could alternatively add $\alpha \sqrt{k/2-1}$ copies of the minimum spanning tree as in Karger's algorithm). 
To fix cuts $S$ of type (ii), we employ the following procedure: for any tree $T_i$ where $\delta(S)_{T_i}=1$ and $S$ is of type (ii), we add one extra copy of the unique edge of $T_i$ in $\delta(S)$. This procedure is in line 5 to line 9 of the algorithm. Let $F$ be the set of edges added in this step; then 
the output of our algorithm is $T^* \uplus B \uplus F$ as in line 10.

Now we analyze \cref{alg:warmup}. 
\begin{theorem}[Approximation Ratio for \Cref{alg:warmup}]\label{thm:main-warmup}
\Cref{alg:warmup} outputs a (weighted) $k$-ECSM with approximation factor (at most)  $1+\sqrt{\frac{8\ln k}{k}}$.
\end{theorem}

We begin by showing that the output of \cref{alg:warmup} is $k$-edge connected (in $G$) with probability 1.

\begin{lemma}[$k$-Connectivity of the Output] \label{lem:connectivity}
    For any $\alpha \ge 0$, the output of \cref{alg:warmup}, $F\uplus B\uplus T^*$ is a $k$-edge connected subgraph of $G$. 
\end{lemma}
\begin{proof}
    Fix spanning trees $T_1,\cdots, T_{k/2}$ in $G^0$ and a cut $S$
    where $(u_0,v_0)\not\in \delta(S)$.  We show that $\delta(S)_{T^*\uplus F\uplus B}\geq k$. 
    If $\delta(S)_{T^*}\geq k-\alpha\sqrt{k/2-1}$, then since $B$ has $\alpha\sqrt{k/2-1}$ copies of the minimum spanning tree, $\delta(S)_{T^*\uplus B}\geq k$ and we are done.
    Otherwise $\delta(S)_{T^*}<k-\alpha\sqrt{k/2-1}$. Then, we know that for any tree $T_i$, 
    if $\delta(S)_{T_i}=1$, since $(u_0, v_0)\not\in \delta(S)_{T_i}$,   $F$  has one extra copy of the unique edge of $T_i$ in $\delta(S)$. 
    Therefore, including those cases where an extra copy of the edge $e$ is added, each $T_i$ has at least two edges in $\delta(S)$, so $\delta(S)_{T^*\uplus F}\geq 2 \cdot \frac{k}{2} \ge k$ as desired since there are $\frac{k}{2}$ spanning trees $T_i$.
\end{proof}

To bound the expected cost of our rounded solution, we use the concentration property of $\lambda$-uniform trees on edges of $T^*$ to show the probability that any fixed cut $\delta(s)$ is in type (ii), i.e. $\delta(S) < k - \alpha \sqrt{k/2-1}$, is exponentially small in $\alpha$, i.e. $\leq e^{-\alpha^2/2}$, even if we condition on $\delta(S)_{T_i} = 1$ for a single tree $T_i$.

In our algorithm we sample $k/2$ trees $T_1,\dots,T_{k/2}$. The following definition will be useful in this section as well as in \cref{sec:alg-analysis}. Note it is important to separate the case in which $(u_0,v_0) \in \delta(S)$ for a cut $S$ because in this event, $x^0(\delta(S))$ may be as small as $k/2$, in which case our analysis is not valid. However, since the $(u_0,v_0)$ edge has cost 0, we need not worry about such cuts since they can be trivially satisfied by adding many copies of this edge.
\begin{definition}[$\mathcal{E}^i_e$]\label{def:calE}
For a tree $T_i$ sampled in \cref{alg:warmup} and an edge $e$, we define $\mathcal{E}^i_e$ to be the event that $e\in T_i\wedge (u_0,v_0)\notin \delta(S_{T_i}(e))$.
\end{definition}

\begin{lemma} \label{lem:prob_augment}
For any $0 \leq \alpha \leq \sqrt{k}$, $1\leq i\leq k/2$, and any $e\in E$, 
	\begin{align*}
		\P{\delta(S_{T_i}(e))_{T^*} \leq k - \alpha \sqrt{k/2-1}  \mid  \mathcal{E}^i_e} \leq e^{-\alpha^2/2}.
	\end{align*}
	where the randomness is over spanning trees $T_1,\cdots, T_{i-1}, T_{i+1}, \cdots, T_{k/2}$ independently sampled from $\mu_\lambda$.
\end{lemma}
\begin{proof}
Condition on tree $T_i$ and the event $\mathcal{E}^i_e$. By  \cref{lem:sum_bernoulli}, for any $1 \le j\leq k/2$ such that $j\neq i$, $\delta(S_{T_i}(e))_{T_j}$ is a $BS( \E{\delta(S_{T_i}(e))_{T_j}})$ random variable, with  $\E{\delta(S_{T_i}(e))_{T_j}} =\frac{2}{k}x(\delta(S_{T_i}(e))) \ge 2$.
    Also, by definition, $\delta(S_{T_i}(e))_{T_i} = 1$ (with probability 1).
    Since $T_1,\cdots, T_{k/2}$ are independently chosen, by \cref{fact:sumBS} 
    the random variable $\delta(S_{T_i}(e))_{T^*}$ is distributed as $BS(q)$ for $q \ge k-1$. 
    Since each $T_j$ has at least one edge in $\delta(S_{T_i}(e))$, $\delta(S_{T_i}(e))_{T^*}\geq k/2$ with probability 1. So,
     by \cref{thm:chernoff}, with $q' = k-1-k/2$, when $0 \leq \alpha \leq \sqrt{k/2-1}$,
    \begin{align*}
        &\P{\delta(S_{T_i}(e))_{T^*} < k - \alpha \sqrt{k/2-1} \mid  \mathcal{E}^i_e }
        \\
        &= \P{\delta(S_{T_i}(e))_{T^*}-k/2 < k/2 - \alpha \sqrt{k/2-1} \mid  \mathcal{E}^i_e} \\
	 	&\leq  e^{-\frac{(\alpha/\sqrt{k})^2 (k/2-1)}{2}} = e^{-\alpha^2 / 2}.
    \end{align*}  
    Averaging over all realizations of $T_i$ satisfying the required conditions proves the lemma.
\end{proof}

\begin{proof}[Proof of \cref{thm:main-warmup}]
Let $x$ be an optimum solution of LP \eqref{eq:k_ecsmlp}.
Since the output of the algorithm is always $k$-edge connected we just need to show $\E{c(F\cup T^*\cup B)}\leq \left(1+\sqrt{\frac{8\ln k}{k}}\right)c(x)$.
    By linearity of expectation, 
    \begin{align*}
        \E{c(T^*)} &= \sum_{i\in [\frac{k}{2}]} \E{c(T_i)} =\frac{k}{2} \sum_{e\in E} c(e) \PP{\mu_\lambda}{e} = \frac{k}{2}\sum_{e\in E} c(e) \cdot  \frac{2}{k}\cdot x_e =c(x),
    \end{align*}
    where for simplicity we ignored the $1+2^{-n}$ loss in the marginals. 
On the other hand, since by \cref{fact:opt_in_sp_polytope}, $\frac{2x}{k}$ is in the spanning tree polytope of $G^0$,  $c(B) \leq \frac{2c(x)}{k} \cdot\alpha\sqrt{k/2-1}\leq \frac{\alpha c(x)
}{\sqrt{k/2}}$.
It remains to bound the expected cost of $F$. By \cref{lem:prob_augment},
\begin{align*}
        \E{c(F)} &= \sum_{e\in E} c(e)\sum_{i=1}^{k/2} \P{\mathcal{E}^i_e} \P{\delta(S_{T_i}(e))_{T^*}<k-\alpha\sqrt{k/2-1}  \mid \mathcal{E}^i_e}  \\
        &\leq \sum_{e\in E} c(e) x_e e^{-\alpha^2/2} 
        \leq e^{-\alpha^2/2} c(x).
    \end{align*}
  Putting these together we get,
		$\E{c(T^* \cup B \cup F)} \leq (1 +\alpha/\sqrt{k/2} + e^{-\alpha^2/2})  c(x).$
Setting $\alpha = \sqrt{\ln\left( \frac{k}{2}\right)}$ finishes the proof.
\end{proof}

\section{Improved Algorithm and Proof of Main Theorem}\label{sec:alg-analysis}

We now introduce our main algorithm  that has an approximation ratio of $1 + O(\frac{1}{\sqrt{k}}))$.
Let $x^0$ be an optimal solution of LP \eqref{eq:k_ecsmlp} extended to $G^0$ as above. Our algorithm is given in \cref{alg:main}. Note for convenience we drop the ceiling in the expression $\frac{k}{2} + \alpha\sqrt{k}$ in all that follows.

\begin{algorithm}[htb]\caption{Algorithm for $k$-ECSM with Approximation Ratio $1 + O(\frac{1}{\sqrt{k}}))$}\label{alg:main}
\begin{algorithmic}[1]
    \State Let $x^0$ be an optimum solution of \eqref{eq:k_ecsmlp} extended to the graph $G^0$ as described above. \label{line:opt_ecss}
    \State Find  weights $\lambda : E^0 \to \R_{\ge 0}$ such that for any $e\in E^0$, $ \PP{\mu_{\lambda}}{e} \le  \frac{2}{k}\cdot x_e^0 \cdot (1 + 2^{-n})$.
    %
    \State Initialize $F \leftarrow \emptyset$.
    \State Sample $k/2 + \alpha \sqrt{k}$ spanning trees $T_1,\cdots,T_{k/2+\alpha \sqrt{k}}\sim\mu_\lambda$ (in $G^0$) independently and let $T^* \leftarrow T_1 \uplus \dots \uplus T_{k/2 + \alpha \sqrt{k}}$. 
    \label{line:sample_st}
    \State Let $\mathcal{S} \leftarrow \{S_{T_i}(e) : i\in [ \frac{k}{2}+\alpha \sqrt{k}],e\in T_i, (u_0,v_0) \not\in \delta(S)\}$. \Comment{$\mathcal{S}$ is the set of one-cuts (see \cref{def:one-cut}) of $T_i \in T^*$.} \label{line:onecut}
    \For{$S\in \mathcal{S}$}
        \State  $P(S):=\uplus_{i=1}^{k/2+\alpha \sqrt{k}} \{e\in T_i: S_{T_i}(e)=S\}$  \Comment{$P(S)$ is the multi-set of $e \in T^*$ with one-cut $S$.}
        \If{$\delta(S)_{T^*} < k$} 
            \For {$j=1$ to $k- \delta(S)_{T^*}$}
                \State Sample an edge from $P(S)$ uniformly at random and add into $F$. \label{line:augment}
            \EndFor
        \EndIf
    \EndFor
    \State {\bf Return} $T^* \uplus F$.
\end{algorithmic}
\end{algorithm}

\mainthm*

We remark that we may assume $k \ge 100$ without loss of generality because for smaller values of $k$ our guarantee is worse than Christofides' algorithm.

\begin{lemma}[$k$-Edge Connectivity of the Output]\label{lem:connectivity}
    The output of \cref{alg:main}, $F\uplus T^*$ is a $k$-edge connected subgraph of $G$. 
\end{lemma}
\begin{proof}
First, note that for every set $S \subset V$ in $G$, the corresponding cut in $G_0$ has $u_0,v_0$ on the same side. Therefore, we may restrict our attention to sets $S \subset V$ such that $u_0,v_0 \not\in S$. However for such an $S$, line 9 of the above algorithm ensures $\delta(S)_{T^* \uplus F} \ge k$, which completes the claim.
\end{proof}

\begin{lemma}[Variance Upper Bound of Cuts in a Random Spanning Tree]\label{lem:variance-bound}
Let $\mu_{\lambda}$ be the max-entropy distribution in \cref{alg:main}. 
For any $0\leq p \leq 1$, any $\epsilon \geq 0$ and any $S\subseteq V$ such that $\PP{T\sim \mu_{\lambda}}{\delta(S)_T = 1} = p$ and $\EE{T\sim \mu_{\lambda}}{\delta(S)_T} = 2 + \epsilon$, we have $\VV{T\sim \mu_{\lambda}}{\delta(S)_T} \leq 4p + 3\epsilon$.
\end{lemma}

\begin{proof}
    By \cref{thm:max-entropy}, $\delta(S)_T$ is distributed as a BS random variable with $\EE{T\sim \mu_{\lambda}}{\delta(S)_T} = 2 + \epsilon$ and $\PP{T\sim \mu_{\lambda}}{\delta(S)_T \geq 1} = 1$. Hence we can write $\delta(S)_T = 1 + X_1 + \cdots + X_m$ for some integer $m \geq 2$ \footnote{We remark that the case for $m = 1$ is trivial.}, where $X_1,\cdots,X_m$ are independent Bernoulli random variables with success probabilities $b_1 \geq b_2 \geq \dots \geq b_m$. Then from the assumption, $\sum_{i=1}^m b_i = 1 + \epsilon$. By \cref{fact:sumBS}, we have 
    \begin{align*}
        \VV{T\sim \mu_{\lambda}}{\delta(S)_T} = \Var{\sum_{i=1}^m X_i} = \sum_{i=1}^m b_i(1-b_i).
    \end{align*}
    
    If $4p \geq 1-2\epsilon$, then we have
    \begin{align*}
        \VV{T\sim \mu_{\lambda}}{\delta(S)_T}  =  \sum_{i=1}^m b_i(1-b_i)  \leq  \sum_{i=1}^m b_i = 1 + \epsilon \leq 4p + 3\epsilon.
    \end{align*}

    Otherwise, $4p < 1 - 2\epsilon$. Notice that
    \begin{align*}
        p &= \Pr[\forall i, X_i = 0] = (1 - b_1) \prod_{i=2}^m (1 - b_i) \geq (1 - b_1) (1 - \sum_{i=2}^m b_i)  = (1 - b_1) \cdot (b_1 - \epsilon).
    \end{align*}
    where the fourth step comes from Weierstrass Inequality, and the last step comes from $\sum_{i=1}^m b_i = 1 + \epsilon$.
    This gives $b_1 \leq \frac{1}{2}(1 + \epsilon - \sqrt{(1 - \epsilon)^2 - 4p})$ or $b_1 \geq \frac{1}{2}(1 + \epsilon + \sqrt{(1 - \epsilon)^2 - 4p})$. Since $4p < 1 - 2\epsilon \leq (1 - \epsilon)^2$, the solutions for $b_1$ are well-defined. By \cref{fact:small-bern} we have $b_1 \geq \frac{1}{2}(1 + \epsilon)$, so $b_1 \geq \frac{1}{2}(1 + \epsilon +\sqrt{(1 - \epsilon)^2 - 4p}) \geq 1 - 2p - \frac{\epsilon}{2}$ (using the square root inequality $\sqrt{1-x} \ge 1-x$ for $0 \le x \le 1$).
    
    Therefore, $\VV{T\sim \mu_{\lambda}}{\delta(S)_T}$ is upper-bounded by:
    \begin{align*}
        \VV{T\sim \mu_{\lambda}}{\delta(S)_T} &= \sum_{i=1}^m b_i (1 - b_i) \leq b_1 (1 - b_1) + \sum_{i=2}^m b_i \\ &=  b_1 (1 - b_1) + (1 + \epsilon - b_1) \\
        &= 1 + \epsilon - b_1^2 
        \leq 1 + \epsilon - (1 - 2p - \frac{\epsilon}{2})^2 
        \leq 4p + 3\epsilon .
    \end{align*}

\end{proof}


As mentioned in \cref{sec:overview}, the following lemma is the key to analyzing \cref{alg:main}. Roughly speaking, it says that the probability a cut is ``bad," i.e. has fewer than $k-\alpha \sqrt{k}$ edges in $T^*$, is exponentially small in the probability that $\delta(S)_{T} = 1$ for $T \sim \mu_\lambda$.

\begin{lemma}[Expected Augmentation of a Cut]\label{lem:aug-large-cut}
    For any $k \ge 100$ and integer $\alpha \ge 1$ let $\mu_{\lambda}$ be the max-entropy distribution and $T^*$ be the union of $\frac{k}{2}+\alpha \sqrt{k}$ random spanning trees sampled from $\mu_{\lambda}$ in \cref{alg:main}. 
    Then for 
    any $S\subseteq V$, 
    \begin{align*}
        \EE{T^*}{ \left( k - \delta(S)_{T^*} \right)^+} \leq 1.8\sqrt{k}\exp\left(\frac{-0.6\alpha}{\max\{k^{-1/2},\P{\delta(S)_T=1}\}}\right)
    \end{align*}
\end{lemma}

\begin{proof}
    We can write the expectation as 
    \begin{align}\label{eq:aug-large-cut-sum}
         \EE{T^*}{ ( k -  \delta(S)_{T^*})^+}
        \le  & \sum_{i= 1}^{k} \PP{T^*}{\delta(S)_{T^*} \leq k - i}   \nonumber \\
        \leq  & \sum_{i= 0}^{\sqrt{k}} \sum_{j=1}^{\sqrt{k}}  \PP{T^*}{\delta(S)_{T^*} \leq k - (i\cdot \sqrt{k} + j)}  \nonumber  \\
        \leq  & \sum_{i= 2\alpha}^{\sqrt{k}} \sqrt{k} \cdot \PP{T^*}{\delta(S)_{T^*} \leq k - (i-2\alpha ) \sqrt{k}}, 
    \end{align}
where we reindex for convenience in the following argument. Define $\beta\geq 0$  such that 
$x(\delta(S))=k+\beta\sqrt{k}$, or equivalently that $\E{\delta(S)_T} = 2(1+\beta/\sqrt{k})$.
By  \cref{lem:variance-bound}, we have
$$\VV{T^*}{\delta(S)_{T^*}} \le (4p+6\beta/\sqrt{k})(k/2+\alpha \sqrt{k})
$$
where $p=\P{\delta(S)_T=1}$. Also, notice,
$$\E{\delta(S)_{T^*}} \ge (k/2+\alpha \sqrt{k})\E{\delta(S)_T}=k+(2\alpha+\beta)\sqrt{k}+2\alpha\beta.$$ 
Therefore, by Bernstein's inequality, for $i \ge 2\alpha$, we have that $\PP{T^*}{\delta(S)_{T^*} \leq k - (i-2\alpha)\sqrt{k}}$ is equal to $$\PP{T^*}{\delta(S)_{T^*} \leq \E{\delta(S)_{T^*}} - (i+\beta)\sqrt{k}-2\alpha \beta}$$
Applying \cref{thm:bernstein}, this is at most 
    \begin{align*} \exp \left( -\frac{(i+\beta)^2k+ 4\alpha\beta(i+\beta)\sqrt{k} }{2(2kp+3\beta\sqrt{k} + 6\alpha\beta +4\alpha p\sqrt{k} +  (i+\beta)\sqrt{k} /3+ 2\alpha \beta/3)} \right)
    \end{align*}
 Now, note that using that $\beta\geq 0$ and the mediant inequality (namely, that for $A,B,C,D \ge 0$ we have $\frac{A+B}{C+D}\geq \min\{A/C, B/D\}$), we can upper bound the term inside the $\exp$ by  
 $$-\min\left\{\frac{i^2k}{4kp + 8\alpha p\sqrt{k} +2i\sqrt{k}/3}, \frac{  2ik + 4\alpha i\sqrt{k}}{40\alpha/3 + 20\sqrt{k}/3} \right\}$$ 
 Therefore, we can bound this probability by
 \begin{align*}
&\leq \exp\left(-\frac{1}{\max\{k^{-1/2},p\}}\min\left\{\frac{i^2}{4+8\alpha/\sqrt{k}+2i/3}, \frac{2i + 4\alpha i/\sqrt{k}}{  20/3+ 40\alpha/3\sqrt{k}}\right\}\right)\\
&\underset{i\geq 2\alpha,\alpha\geq 1, k\geq 100}{\leq} \exp\left(\frac{-0.3i}{\max\{k^{-1/2},p\}}\right) 
    \end{align*}
Therefore,
\begin{align*}
   \EE{T^*}{ ( k - \delta(S)_{T^*})^+} &\leq \sqrt{k} \sum_{i=2\alpha}^{\sqrt{k}} \exp\left(\frac{-0.3i}{\max\{k^{-1/2},p\}}\right) \\
   &\underset{p,k^{-1/2}\leq 1/e}{\leq} \sqrt{k} \exp\left(\frac{-0.6\alpha}{\max\{k^{-1/2},p\}}\right) \sum_{i=0}^\infty e^{-0.3ei} \\
   &\leq 1.8\sqrt{k}\exp\left(\frac{-0.6\alpha}{\max\{k^{-1/2},p\}}\right)
\end{align*}

\end{proof}

Given the above lemma, the expected cost of $F$ follows from a relatively straightforward calculation: 

\begin{lemma}[Expected Payment of an Edge for Augmentation]\label{lem:aug-main}
    For any $k \ge 100$ and integer $\alpha \ge 1$, let $T^*$ be the union of $\frac{k}{2}+\alpha \sqrt{k}$ random spanning trees $T_1, \cdots, T_{k/2+\alpha \sqrt{k}}$ in \cref{alg:main}. For any solution $x$ to LP \eqref{eq:k_ecsmlp}, 
    \begin{align*}
        \EE{T^*}{c(F)} \leq \left(1+\frac{2\alpha}{\sqrt{k}}\right)\left(\frac{7.2e}{\sqrt{k}} e^{-0.6\alpha e} + e^{-\sqrt{k}/2}\right)c(x)
    \end{align*}
    where $F$ is as defined in \cref{alg:main}. 
\end{lemma}

\begin{proof}
    Fix any $i\in [\frac{k}{2} + \alpha \sqrt{k}]$, 
     condition on $T_i$, fix an edge $e\in T_i$ such that $u_0,v_0\notin S_{T_i}(e)$. (If $v_0 \in S_{T_i}(e)$, then this is not a cut in the original graph $G$, so there is nothing to prove). Let $S = S_{T_i}(e)$ and let $p = \PP{T\sim \mu_{\lambda}}{\delta(S)_T = 1}$.
    Recall $P(S):=\uplus_{j=1}^{k/2+\alpha \sqrt{k}}\{f\in T_j: S_{T_j}(f)=S\}$ denotes the multi-set of edges $f\in T_j$ for all $1\leq j\leq k/2+\alpha \sqrt{k}$, such that $S_{T_j}(f) = S$.  
    
    Let $X_{T_i,e}$ be the number of times that edge $e$ from tree $T_i$ is sampled in \cref{line:sample_st}, line 10.
    We will prove that, letting $\mathcal{E}^i_e$ denote the event $e \in T_i, (u_0,v_0)\notin \delta(S_{T_i}(e))$,
    \begin{equation}\label{eq:Ece}
        \E{X_{T_i,e} \mid \mathcal{E}_e} 
        \leq \frac{7.2e}{\sqrt{k}} e^{-0.6\alpha e} + e^{-\sqrt{k}/2}
    \end{equation}
    Then, to prove the lemma,
    \begin{align*}\label{eq:main-cF}
        \E{c(F)} = ~ & \sum_{e\in E} c(e)\sum_{i=1}^{k/2+\alpha \sqrt{k}} \P{\mathcal{E}^i_e} \cdot \E{X_{T_i,e} \mid \mathcal{E}^i_e} \nonumber \\
        &\underset{\eqref{eq:Ece}}{\leq} \sum_{e\in E} c(e)\sum_{i=1}^{k/2+\alpha \sqrt{k}} \P{\mathcal{E}^i_e} \left(\frac{7.2e}{\sqrt{k}} e^{-0.6\alpha e} + e^{-\sqrt{k}/2}\right) \\
        &\le \sum_{e\in E} c(e)\left(\frac{k}{2}+\alpha \sqrt{k}\right) \cdot \frac{2}{k} x_e \cdot  \left(\frac{7.2e}{\sqrt{k}} e^{-0.6\alpha e} + e^{-\sqrt{k}/2}\right)\\  &= \left(1+\frac{2\alpha}{\sqrt{k}}\right)\left(\frac{7.2e}{\sqrt{k}} e^{-0.6\alpha e} + e^{-\sqrt{k}/2}\right)c(x)
    \end{align*}
    
    In the rest of the proof we show \eqref{eq:Ece}. 
    First, observe that
    \begin{align*}
        \E{X_{T_i,e} \mid T_i, e\in T_i, \delta(S)_{T^*}, |P(S)|} =  & \frac{( k- \delta(S)_{T^*})^+ }{ |P(S)|}.
    \end{align*}
    So, it is enough to upper bound the expected value of the RHS conditioned on $T_i, e\in T_i$.
    Let $T^*_{-i} = T^*\smallsetminus T_i$ and $P(S)_{-i}:=P(S)\smallsetminus T_i$ and note that $|P(S)|=|P(S)_{-i}|+1$ and $\delta(S)_{T^*}=\delta(S)_{T^*_{-i}}+1$. Define $Y_{e,T_i} = \frac{( k - \delta(S)_{T^*_{-i}}+1)^+ }{ |P(S)_{-i}|+1}$.
    Note that $\E{X_{T_i,e} | T_i, e\in T_i} = \E{Y_{e,T_i}}$.
    We drop the subscript of $Y$ for readability.
    
    First, assume $p \ge 2/k$. Then, we write,
    \begin{equation}\label{eq:EYopenup}
    	    \begin{aligned}
        \E{Y} &= \P{|P(S)_{-i}| \geq \frac{pk}{4}} \cdot \E{Y \mid  |P(S)_{-i}| \geq \frac{pk}{4}} \\
         &+ \P{|P(S)_{-i}| < \frac{pk}{4}} \cdot \E{Y \mid |P(S)_{-i}| < \frac{pk}{4}}
    \end{aligned}
    \end{equation}

We upper bound each term in the RHS separately. 
\begin{align*}
    &\P{|P(S)_{-i}| \geq \frac{pk}{4}} \E{Y \mid |P(S)_{-i}| \geq \frac{pk}{4}}\\  &\leq \frac{4}{pk}\P{|P(S)_{-i}| \geq \frac{pk}{4}} \E{(k-\delta(S)_{T^*_{-i}})^+ \mid |P(S)_{-i}| \geq \frac{pk}{4}} \\
    &\leq  \frac{4}{pk}\E{(k-\delta(S)_{T^*_{-i}})^+} \leq \frac{7.2}{p\sqrt{k}}\exp\left(\frac{-0.6\alpha }{\max\{p,k^{-1/2}\}}\right)
\end{align*}
Now, if $p>k^{-1/2}$, the RHS is maximized when $p=1/e$ since $\frac{1}{p} e^{-c/p}$ is an increasing function of $p$ for $0 \le p \le 1/e$; note $p \le 1/e$ (by \cref{cor:1/e}) and $\alpha \ge 1$. We obtain a bound of $\frac{7.2e}{\sqrt{k}} e^{-0.6\alpha e}$.

Otherwise, the maximum is achieved by $k^{-1/2}$. Using $p \ge 2/k$, the above expression is at most $3.6\sqrt{k}e^{-0.6\alpha k^{1/2}}$. 
So, the first term is at most $\frac{7.2e}{\sqrt{k}} e^{-0.6\alpha e}$ for $k\geq 100$.

Next, we bound the second term of \eqref{eq:EYopenup}.
First notice $Y\leq 1$ with probability 1; this is because if there are exactly $\ell$ trees which have $(S,\overline{S})$ as a one-cut then, $\delta(S)_{T^*}\geq k+2\alpha \sqrt{k}-\ell$ whereas $|P(S)|=\ell$. Furthermore, $Y\neq 0$ only when $|P(S)_{-i}|\geq 2\alpha \sqrt{k}\geq 2\sqrt{k}$ (for $\alpha\geq 1$). Therefore,
\begin{align*}
    &\P{|P(S)_{-i}| < \frac{pk}{4}} \cdot \E{Y \mid |P(S)_{-i}| < \frac{pk}{4}} \\ &\leq \P{|P(S)_{-i}| < \frac{pk}{4}} \P{|P(S)_{-i}|\geq2\sqrt{k} \mid |P(S)_{-i}| < \frac{pk}{4}}\\
     &= \P{2\sqrt{k}\leq |P(S)_{-i}|\leq \frac{pk}{4}} \leq e^{-pk/16} \leq e^{-\sqrt{k}/2}
\end{align*}
To see the last two inequalities  notice we must have $p\geq 8k^{-1/2}$ or this event cannot occur.
 Therefore, since $\E{|P(S)_{-i}|} = p(\frac{k}{2}+\alpha\sqrt{k}-1)$ the inequality follows  by an application of the Chernoff bound (\cref{thm:chernoff}). 
 
Putting these two terms together, if $p > 2/k$, 
$$\E{X_{T_i,e} \mid T_i, e \in T_i} = \E{Y} \le \frac{7.2e}{\sqrt{k}} e^{-0.6\alpha e} + e^{-\sqrt{k}/2}$$ 
 
Otherwise, suppose $p \le 2/k$. Then since $\E{|P(S)|_{-i}} \le 1+2\alpha k^{-1/2}$, by \cref{thm:chernoff} (using $\alpha \ge 1$),
\begin{align*}
	\P{|P(S)_{-i}|>(1+2\alpha\sqrt{k})(1+2\alpha k^{-1/2})} &\leq e^{-\frac{4\alpha^2k(1+2\alpha k^{-1/2})}{2+2\alpha\sqrt{k}}} \underset{k\geq 10,\alpha\geq 1}{\le} e^{-\sqrt{k}}.
\end{align*}

Since $Y \le 1$ as observed above, we obtain 
$$\E{Y} \leq  \P{\delta(S)_{T^*_{-i}} \le k} \le \P{|P(S)|_{-i} \ge 2\alpha\sqrt{k}} \le e^{-\sqrt{k}}$$ which gives \eqref{eq:Ece}. Therefore we can bound $\E{Y}$ by $\frac{7.2e}{\sqrt{k}} e^{-0.6\alpha e} + e^{-\sqrt{k}/2}$ for all values of $p$.
\end{proof}

\begin{proof}[Proof of \cref{thm:main}]
    Let $x$ be the optimum solution of \eqref{eq:k_ecsmlp}. From \cref{lem:connectivity}, the output of \cref{alg:main} is always $k$-edge connected. Thus it suffices to show that $\E{c(T^* \uplus F)}\leq (1+ \frac{5.06}{\sqrt{k}})c(x)$. By linearity of expectation, 
    \begin{align*}
        \E{c(T^*)} &= \sum_{i\in [\frac{k}{2}+\alpha\sqrt{k}]} \E{c(T_i)} =(\frac{k}{2}+\alpha\sqrt{k}) \sum_{e\in E} c(e) \PP{\mu_\lambda}{e} \\ &= (\frac{k}{2}+\alpha\sqrt{k})\sum_{e\in E} c(e) \cdot  \frac{2}{k}\cdot x_e =(1+\frac{2\alpha}{\sqrt{k}})c(x),
    \end{align*}
    where for simplicity we ignored the $1+2^{-n}$ loss in the marginals. 
  
    Therefore, by \cref{lem:aug-main}, $\E{c(T^* \uplus F)}$ is at most
    \begin{align*}
        &\leq c(x) \cdot \left(1 + \frac{2\alpha}{\sqrt{k}} + \left(1+\frac{2\alpha}{\sqrt{k}}\right)\left(\frac{7.2e}{\sqrt{k}} e^{-0.6\alpha e} + e^{-\sqrt{k}/2}\right) \right) \leq c(x) \cdot \left(1+\frac{5.06}{\sqrt{k}}\right). 
    \end{align*}
    as desired, where in the last inequality we use $k \ge 100$ and set $\alpha = 2$.
\end{proof}

\section{Conclusion}
We remark that the approximation factor $1+O(1/\sqrt{k})$ is tight for any algorithm that starts by sampling $O(k)$ spanning trees independently from the max-entropy distribution and then fixes the union by adding edges. For a tight example, consider a complete graph with a unit metric on the edges and let $x$ be uniform across all edges. In such a case, the max-entropy distribution $\mu_\lambda$ will be the uniform distribution over all spanning trees of a complete graph. A simple analysis shows that  every vertex will have degree $k-\sqrt{k}$ in $T^*$ with constant probability. Therefore, to fix $T^*$ we need to add at least $\Omega(n\sqrt{k})$ edges. 

It still remains open if the integrality gap of the LP is indeed $1+O(1/k)$ or if there is an approximation algorithm with approximation factor $1+O(1/k)$. It would also be interesting to find the optimal constant for \cref{alg:main}.

\printbibliography

@misc{KKO21b,
      title={A (Slightly) Improved Bound on the Integrality Gap of the Subtour LP for TSP}, 
      author={Anna Karlin and Nathan Klein and Shayan {Oveis Gharan}},
      year={2021},
      eprint={2105.10043},
      archivePrefix={arXiv},
      primaryClass={cs.DS}
}

@inproceedings{CVZ10,
	Author = {Chandra Chekuri and Jan Vondr{\'a}k and Rico Zenklusen},
	Booktitle = {FOCS},
	Pages = {575-584},
	Title = {Dependent Randomized Rounding via Exchange Properties of Combinatorial Structures},
	Year = {2010}}

@inproceedings{CTZ21,
  author    = {Federica Cecchetto and
               Vera Traub and
               Rico Zenklusen},
  editor    = {Samir Khuller and
               Virginia Vassilevska Williams},
  title     = {Bridging the gap between tree and connectivity augmentation: unified
               and stronger approaches},
  booktitle = {STOC},
  pages     = {370--383},
  publisher = {{ACM}},
  year      = {2021},
}

@article{TZ21,
  author    = {Vera Traub and
               Rico Zenklusen},
  title     = {A Better-Than-2 Approximation for Weighted Tree Augmentation},
  booktitle={FOCS},
  year = {2021},
  note={to appear},
  
}

@inproceedings{BGA20,
  author    = {Jaroslaw Byrka and
               Fabrizio Grandoni and
               Afrouz Jabal Ameli},
  editor    = {Konstantin Makarychev and
               Yury Makarychev and
               Madhur Tulsiani and
               Gautam Kamath and
               Julia Chuzhoy},
  title     = {Breaching the 2-approximation barrier for connectivity augmentation:
               a reduction to Steiner tree},
  booktitle = {STOC},
  pages     = {815--825},
  publisher = {{ACM}},
  year      = {2020},
}

@article {FJ82,
	title = {On the relationship between the biconnectivity augmentation and traveling salesman problem},
	journal = {Theoretical Computer Science},
	volume = {19},
	year = {1982},
	pages = {189 - 201},
	author = {Fredrickson,G. N. and J\'aJ\'a, Joseph F.}
}

@inproceedings{FGKS18,
author = {Fiorini, Samuel and Gro{\ss}, Martin and K{\"o}nemann, Jochen and Sanit{\`a}, Laura},
title = {Approximating Weighted Tree Augmentation via Chvátal-Gomory Cuts},
booktitle = {SODA},
year={2018},
pages = {817-831},
}

@article{GG08,
  title={Iterated rounding algorithms for the smallest k-edge connected spanning subgraph},
  author={H. Gabow and S. Gallagher},
  journal={SIAM J. Comput.},
  year={2008},
  volume={41},
  pages={61-103}
}

@article{KR96,
author={S. Khuller and B. Raghavachari}, 
title={Improved approximation algorithms for uniform connectivity problems}, 
journal={J Algorithms}, 
volume={21}, 
year={1996}, 
pages={434--450},
}

@inproceedings{OSS11,
  author    = {Shayan {Oveis Gharan} and
               Amin Saberi and
               Mohit Singh},
  title     = {A Randomized Rounding Approach to the Traveling Salesman Problem},
  booktitle = {FOCS},
  pages     = {550--559},
  publisher = {{IEEE} Computer Society},
  year      = {2011},
}

@inproceedings{LOS12,
  author    = {Bundit Laekhanukit and
               Shayan {Oveis Gharan} and
               Mohit Singh},
  editor    = {Artur Czumaj and
               Kurt Mehlhorn and
               Andrew M. Pitts and
               Roger Wattenhofer},
  title     = {A Rounding by Sampling Approach to the Minimum Size k-Arc Connected
               Subgraph Problem},
  booktitle = {ICALP},
  series    = {Lecture Notes in Computer Science},
  volume    = {7391},
  pages     = {606--616},
  publisher = {Springer},
  year      = {2012},
}

@article{Kar99,
author={D. Karger}, 
title={Random sampling in cut, flow, and network design problems}, 
journal={Math OR},
volume={24}, 
year={1999}, 
pages={383--413},
}

@article{Gab05,
author={H. Gabow}, 
title={An improved analysis for approximating the smallest k-edge connected spanning subgraph of a multi-graph}, 
journal={SIAM J Disc Math}, 
volume={19}, 
year={2005}, 
pages={1--18},
}

@InProceedings{BCCGISW20,
  author =	{Sylvia Boyd and Joseph Cheriyan and Robert Cummings and Logan Grout and Sharat Ibrahimpur and Zolt{\'a}n Szigeti and Lu Wang},
  title =	{{A 4/3-Approximation Algorithm for the Minimum 2-Edge Connected Multisubgraph Problem in the Half-Integral Case}},
  booktitle =	{APPROX/RANDOM},
  pages =	{61:1--61:12},
  year =	{2020},
  volume =	{176},
  editor =	{Jaros{\l}aw Byrka and Raghu Meka},
  publisher =	{Schloss Dagstuhl--Leibniz-Zentrum f{\"u}r Informatik},
}

@article{BFS16,
author={Sylvia Boyd and Yao Fu and Yu Sun},
title={A 5/4-approximation for subcubic 2EC using circulations and obliged edges}, 
journal={Discrete Applied Mathematics}, 
volume={209},
pages={48--58}, 
year={2016},
}

@InProceedings{CR98,
author="Carr, Robert
and Ravi, R.",
editor="Bixby, Robert E.
and Boyd, E. Andrew
and R{\'i}os-Mercado, Roger Z.",
title="A New Bound for the 2-Edge Connected Subgraph Problem",
booktitle="IPCO",
year="1998",
publisher="Springer Berlin Heidelberg",
address="Berlin, Heidelberg",
pages="112--125",
}

@article{SV14,
author={Andr\'as Seb\"o and Jens Vygen},
title={Shorter tours by nicer ears: 7/5-Approximation for the graph-TSP, 3/2 for the path version, and 4/3 for two-edge-connected subgraphs},
journal={Combinatorica}, 
volume={34},
number={5},
pages={597--629}, 
year={2014},
}

@article{JT00,
author={J. Cheriyan and R. Thurimella},
title={Approximating minimum- size k-connected spanning subgraphs via matching}, 
journal={SIAM J Comput}, 
volume={30}, 
year={2000}, 
pages={528--560},
}

@article{GGTW09,
  author    = {Harold N. Gabow and
               Michel X. Goemans and
               {\'{E}}va Tardos and
               David P. Williamson},
  title     = {Approximating the smallest \emph{k}-edge connected spanning subgraph
               by LP-rounding},
  journal   = {Networks},
  volume    = {53},
  number    = {4},
  pages     = {345--357},
  year      = {2009},
}

@InProceedings{Pri11,
author="Pritchard, David",
editor="Jansen, Klaus
and Solis-Oba, Roberto",
title="k-Edge-Connectivity: Approximation and LP Relaxation",
booktitle="Approximation and Online Algorithms",
year="2011",
publisher="Springer Berlin Heidelberg",
address="Berlin, Heidelberg",
pages="225--236",
}

@article{FJ81,
author = {Fredrickson,G. N. and J\'aJ\'a, Joseph F.},
journal={SIAM J. Comput.}, 
volume={10},
year={1981},
number={2}, 
pages={270-283},
title={Approximation Algorithms for Several Graph Augmentation Problems},
}

@inproceedings{KKO21,
      title={A (Slightly) Improved Approximation Algorithm for Metric TSP}, 
      author={Anna R. Karlin and Nathan Klein and Shayan {Oveis Gharan}},
      year={2021},
      booktitle={STOC},
      publisher={ACM},
}

@article{GB93,
  author    = {Michel X. Goemans and
               Dimitris Bertsimas},
  title     = {Survivable networks, linear programming relaxations and the parsimonious
               property},
  journal   = {Math. Program.},
  volume    = {60},
  pages     = {145--166},
  year      = {1993}
}

@inproceedings{Edm70,
  Address = {New York, NY, USA},
  Author = {Edmonds, Jack},
  Booktitle = {Combinatorial Structures and Their Applications},
  Pages = {69--87},
  Publisher = {Gordon and Breach},
  Title = {Submodular functions, matroids and certain polyhedra},
  Year = {1970}}

@article{Pit97,
  title={Probabilistic bounds on the coefficients of polynomials with only real zeros},
  author={ Pitman, Jim },
  journal={Journal of Combinatorial Theory, Series A},
  volume={77},
  number={2},
  pages={279-303},
  year={1997},
}

@article{BBL09,
  title={Negative dependence and the geometry of polynomials},
  author={ Borcea, Julius  and  Brändén, Petter and  Liggett, Thomas M. },
  journal={Journal of the American Mathematical Society},
  volume={22},
  number={2},
  pages={521-567},
  year={2009},
}

@article{AGMOS17,
  title={An O (log n /log log n )-Approximation Algorithm for the Asymmetric Traveling Salesman Problem},
  author={ Asadpour, Arash  and  Goemans, Michel X.  and  Mdry, Aleksander  and  Gharan, Shayan Oveis  and  Saberi, Amin },
  journal={Operations Research},
  volume={65},
  number={4},
  year={2017},
}

@article{Hoe56,
  title={On the distribution of the number of successes in independent trials},
  author={Hoeffding, Wassily},
  journal={Annals of Mathematical statistics},
  volume={27},
  number={3},
  pages={713--721},
  year={1956},
  publisher={Institute of Mathematical Statistics}
}



\end{document}